\documentclass[12pt]{article}
\usepackage[usenames]{color}
\usepackage{tabularx,colortbl}
\usepackage{amsmath,amssymb,amsfonts}
\usepackage{epsfig}

\newtheorem{thm}{Theorem}[section]

\newtheorem{remark}{Remark}[section]

\def\sds{\strut \displaystyle}

\newcommand{\FF}{\mathbb{F}}

\def\F{{\mathbb F}}

\def\C{{\mathbb C}}

\def\ri{{\rm i}}
\newcommand{\MS}{\mathbb{S}}

\newtheorem{definition}{Definition}
\newtheorem{example}{Example}

\newcommand{\beq}{\begin{equation}}
\newcommand{\eeq}{\end{equation}}
\newcommand{\beas}{\begin{eqnarray*}}
\newcommand{\eeas}{\end{eqnarray*}}
\newcommand{\bea}{\begin{eqnarray}}
\newcommand{\eea}{\end{eqnarray}}
\newcommand{\bei}{\begin{itemize}}
\newcommand{\eei}{\end{itemize}}
\newcommand{\ben}{\begin{enumerate}}
\newcommand{\een}{\end{enumerate}}
\newcommand{\bet}{\begin{theorem}}
\newcommand{\eet}{\end{theorem}}
\newcommand{\bel}{\begin{lemma}}
\newcommand{\eel}{\end{lemma}}
\newcommand{\bep}{\begin{proposition}}
\newcommand{\eep}{\end{proposition}}
\newcommand{\bed}{\begin{definition}}
\newcommand{\eed}{\end{definition}}
\newcommand{\bec}{\begin{corollary}}
\newcommand{\eec}{\end{corollary}}
\newcommand{\bex}{\begin{example}}
\newcommand{\eex}{\end{example}}

\newcommand{\argmin}{\mathop{\rm arg\min}}

\newcommand{\abs}[1]{\lvert#1\rvert}
\newcommand{\bigabs}[1]{\big\lvert#1\big\rvert}
\newcommand{\Bigabs}[1]{\Big\lvert#1\Big\rvert}

\newenvironment{proof}[1][Proof]{\noindent\textbf{#1.} }{\
  \rule{0.5em}{0.5em}}

\begin{document}

\markboth{{\small\it Compressed Sensing Matrices from Fourier Matrices}} {{\small\it Z. XU and G. Xu}}

\title{Compressed Sensing Matrices from Fourier Matrices }

\author{
 Guangwu Xu \thanks{Department of EE \& CS, University of Wisconsin-Milwaukee,
Milwaukee, WI 53211, USA; e-mail: {\tt gxu4uwm@uwm.edu}. Research
supported in part by the National 973 Project of China (No.
2013CB834205).}  and
Zhiqiang Xu \thanks{Inst. Comp. Math., Academy of Mathematics and Systems Science,
Chinese Academy of Sciences, Beijing, China;
e-mail: {\tt xuzq@lsec.cc.ac.cn}. Zhiqiang Xu was supported  by NSFC grant 11171336 and by the Funds for Creative
Research Groups of China (Grant No. 11021101).}
}

\date{}

\maketitle

\begin{abstract}
The class of Fourier matrices is of special importance in compressed sensing
(CS).
This paper concerns deterministic construction of compressed sensing
matrices from Fourier matrices. By using Katz' character sum estimation, we
are able to design a deterministic procedure to select rows from a
Fourier matrix to form a good compressed sensing matrix for sparse
recovery. The sparsity bound in our construction is similar to that
of binary CS matrices constructed by DeVore which greatly improves
previous results for CS matrices from Fourier matrices. Our approach
also provides more flexibilities in terms of the dimension of CS matrices.
As a consequence, our construction yields an approximately mutually
unbiased bases from Fourier matrices which is of particular interest
to quantum information theory. This paper also contains a useful improvement
to Katz' character sum estimation for quadratic extensions, with an
elementary and transparent proof. Some numerical examples are included.

\end{abstract}


\section{Introduction}
In many practical situations the data  of concern is sparse under suitable representations.
Many problems turn
to be computationally amenable under the sparsity assumption. As a
notable example, it is now well understood that the $\ell_1$
minimization method provides an effective way for reconstructing
sparse signals in a variety of settings.

The goal of compressed sensing (CS) is to recover
a sparse vector $\beta\in\C^N$ from the linear sampling $y=\Phi \beta$ and the
sampling matrix (or compressed sensing matrix) $\Phi\in \C^{m\times N}$. A general model can be of the form $y = \Phi\beta + z$
with $z$ being a vector of errors (noise). In this case,
one needs to approximate the sparse vector $\beta$ from the linear sampling $y$, the compressed sensing matrix $\Phi$ and some information of $z$. However we shall only be
interested in the noiseless case (i.e., $z=0$) in this paper to simplify discussion.

In order to achieve less samples,
one requires that the size of $y$ be much smaller than the dimension of $\beta$, namely $m \ll N$.
A na\"ive approach for solving this problem is to consider
$\ell_0$ minimization, i.e.,
 \begin{equation}\label{P0}
(P_{0}) \quad\quad
\hat \beta = \argmin_{\gamma \in \C^N} \{\|\gamma\|_0 \; \mbox{
  subject to } \; y = \Phi\gamma \}.
\end{equation}
 However this is computationally infeasible. It is then natural to consider the
method of $\ell_1$ minimization which can be viewed as a convex
relaxation of $\ell_0$ minimization. The $\ell_1$ minimization method in
this context is
\begin{equation}\label{P1}
(P_{1}) \quad\quad
\hat \beta = \argmin_{\gamma \in \C^N} \{\|\gamma\|_1 \; \mbox{
  subject to } \; y = \Phi\gamma \}.
\end{equation}
This method has been successfully used as an effective way for reconstructing a
sparse signal in many settings.

Here, a central problem is to construct the compressed sensing matrices so that the solution to $(P_1)$ is the same as that to $(P_0)$.
One of the most commonly used frameworks for the compressed sensing matrices $\Phi$ is the restricted isometry property (RIP) which was
introduced by Cand\`es and Tao \cite{CanTao05}. RIP has been used in the randomized construction
of CS matrices (see for example \cite{BaDaDeWa,CanTao05,CanTao07}).
Another well-known framework in compressed sensing is the mutual incoherence property (MIP) of
Donoho and Huo \cite{DonHuo}. Several deterministic constructions of CS matrices are based on
MIP, e.g., \cite{Dev,LGGZ,xu}.

The main focus of this paper is on the deterministic construction of CS matrices.
In \cite{Dev}, DeVore presented a  deterministic construction of CS matrices using the
mutual incoherence: given a prime number $p$ and an integer $1<n\le p$, an $m\times N$ (with $m=p^2, N = p^{n}$)
binary matrix (i.e., each of its entry  is either $0$ or $1$) $\Phi$ can be found in a deterministic manner
such that whenever
\begin{equation}\label{DeVore_Bound}
k < \frac{\sqrt{m}}{2(n-1)} +\frac{1}2,
\end{equation}
a signal $\beta$ with the sparsity $k$ (i.e., $\beta$ has at most $k$ nonzero components, we also call such a vector $k$-sparse) can be produced by
solving ($P_1$)\footnote{In \cite{Dev}, it was stated that for RIP bound $\delta < 1$, it suffices that $k < \frac{\sqrt{m}}{n-1} +1$.
But for being able to recover $k$-sparse signal, one needs to use (\ref{DeVore_Bound}) based on the discussion in \cite{CaWaXu}. }.
Recently,  Li, Gao, Ge, and Zhang \cite{LGGZ} suggested a deterministic construction of of binary CS matrices
via algebraic curves over finite fields. As stated in \cite{LGGZ}, their construction is more flexible and
slightly improves DeVore's result when $N$ is large (in fact the examples in \cite{LGGZ} indicate that one sees
improvement only when $N\ge \Theta(m^{\sqrt[4]{m}})$).

It is remarked that, in terms of construction via MIP,
the bound (\ref{DeVore_Bound}) is asymptotically optimal because of the
Welch lower bound \cite{welch} (details will be given in section 2). This
means that one needs to require $k = O(\sqrt{m})$.
Using new estimates for sumsets in product sets and
for exponential sums with the products of sets possessing special additive
structure, Bourgain, Dilworth, Ford, Konyagin and Kutzarva \cite{BoDiFoKo}
were able to overcome the natural barrier $k = O(\sqrt{m})$. However
there is a restriction in the construction of \cite{BoDiFoKo},
namely the ratio $\frac{N}m$ is small.

Compared with binary matrices,  the class of Fourier matrices is of great theoretical and practical relevance to compressed sensing.
A notable technique of designing random partial Fourier matrices
 was suggested by Cand\`es and Tao \cite{CanTao06}, and was improved by
Rudelson and Vershynin \cite{RuVe}. Let ${\cal F}^{(N)}$ be the $N \times N$
Fourier matrix whose $(k,j)$-th entry is given by
\[
\left({\cal F}^{(N)}\right)_{k,j} = \exp\left(\frac{2\pi \ri kj}N\right).
\]
For $\Gamma = \frac{1}{\sqrt{N}}{\cal F}^{(N)}$, it was proved in
\cite{RuVe} that if
\begin{equation}\label{rand_condition}
m = O(k\log^4N),
\end{equation}
then a submatrix $\Phi$ consists of $m$ random rows of $\Gamma$
satisfying RIP conditions (which ensures $k$-sparse signal recovery) with high
probability. This means that given $N$ and $m$, the optimal $k$ is
$k= O\left(\frac{m}{\log^4 N}\right)$.

Deterministic construction of CS matrices from Fourier matrices  has also
received recent attention.  In \cite{Xia}, by using the difference set, Xia, Zhou
 and Giannakis constructed a deterministic partial Fourier matrices with size $m\times N$ whose
mutual incoherence constant meets the Welch lower bound (see Section 2). Such
matrices are CS matrices for recovering a
 $k$-sparse signal whenever
$k < \frac{1}{2}\bigg(\sqrt{\frac{(N-1)m}{N-m}}+1\bigg)$.
However, the restriction on the  matrix size ($m$ and $N$) is heavy. In fact,
one can only construct such partial Fourier matrices for some very special
parameters $(m, N)$  with small $\frac{N}m$ (see \cite{Xia}). Recently, Haupt, Applebaum, and Nowak proposed a deterministic procedure to produce a  partial
Fourier matrix
with size $m\times N$ for a large class of $(m,N)$ pairs \cite{HaApNo}.
More specifically, the matrices in \cite{HaApNo} are of size
$m\times N$ with $N$ being a prime and $m\in [N^{\frac{1}{d-1}},  N]$
for some integer $d\ge 2$. This class of CS matrices can be used to recover
$k$-sparse signal with
\begin{equation}\label{determ_condition_2}
k = O\left( m^{\frac{1}{9d^2 \log d}}\right).
\end{equation}

It is noted that there is a significant gap between the bound (\ref{determ_condition_2}) and the (asymptotically optimal) bound (\ref{DeVore_Bound}).
The allowed range of sparsity for the case of a
deterministic partial Fourier matrices is far smaller than that of the binary case.
Constructing partial Fourier matrices deterministically that work for a larger
sparsity range of signals is certainly of particular interest.

The aim of this paper is to construct a partial Fourier matrix which is a CS matrix,
through a deterministic procedure.  Based on a celebrated character sum estimation of Katz \cite{katz},
we are able to obtain a bound for the sparsity $k$ in the case of partial
Fourier matrices
that is similar to  (\ref{DeVore_Bound}). More precisely, we have shown that if
$q=p^a$ is a prime power and $n>1$, setting $N=q^n-1$ or $N=\frac{q^n-1}{p^b-1}=\frac{p^{an}-1}{p^b-1}$($b\mid a$ is required for this case),
then there is a deterministic process to select $m=q$ rows from the Fourier matrix ${\cal F}^{(N)}$ and build a (column normalized)
matrix $\Phi$, such that if
\begin{equation}\label{determ_condition_3}
k < \frac{\sqrt{m}}{2(n-1)} +\frac{1}2,
\end{equation}
then $\Phi$ can be used to reconstruct $k$-sparse signals via $(P_1$).
It is noted that (\ref{determ_condition_3}) also greatly improves (\ref{determ_condition_2}).
The result in this paper can also be used to recover the sparse trigonometric polynomial with a single variable \cite{rauhut1,rauhut}.

In this paper, we also  improve Katz' estimation for
quadratic extension fields with an elementary and transparent approach.
Using this improvement, we are able to construct a CS matrix which is
a partial Fourier matrix, and whose columns are a union of orthonormal
bases. This is a useful construction for sparse representation of signals
in a union of orthonormal
bases which has been a topic of some studies (see, for example
\cite{DonHuo,EladBruck,GhJa,GrNi}). Moreover, this construction produces  an
approximately mutually unbiased bases which is of particular interest
in quantum information theory.
We also
conduct some numerical experiments. The results show that the deterministic
partial Fourier matrices has a better performance over the random partial
Fourier matrices,
provided that these two classes of matrices are of comparable sizes.

This paper is organized as follows. The section  below provides
some necessary concepts and results to be used in our discussion. The
main results are given in section \ref{sec:main_res}.
The discussion of computational issues and
numerical results are contained in the last section.

\section{Background and Preparation}
\label{sec:notation}

\subsection{MIP and Sparse Recovery}
Let $\Phi$ be an $m\times N$ matrix with normalized  column vectors $\Phi_1, \Phi_2, \ldots, \Phi_N$.
Assume that each $\Phi_i$ is of unit (Euclidean) length. The mutual incoherence constant (MIC) is defined as
\[
\mu = \max_{i\neq j}|\langle \Phi_i, \Phi_j \rangle|.
\]
Even though in many situations a small $\mu$ is desired, the following
well-known result of Welch \cite{welch} indicates that
$\mu$ is bounded below
\[
\mu\,\, \geq \,\, \sqrt{\frac{N-m}{(N-1)m}}.
\]

In \cite{DonHuo}, Donoho and Huo gave a
computationally verifiable condition on the parameter $\mu$ that ensures the sparse signal recovery: let $\Phi$
be a concatenation of two orthonormal
matrices. Assume $\beta$ is $k$-sparse. If
\begin{equation}\label{mip}
\mu < \frac{1}{2k-1},
\end{equation}
then the solution $\hat\beta$ for $(P_1)$ is exactly $\beta$.

This result of \cite{DonHuo} was extended to a general matrix $\Phi$ by Fuchs \cite{Fuchs}, and,
Gribonval and Nielsen \cite{GrNi}, both in the noiseless case. For the noisy case of the bounded error,
\cite{DonElaTem,CaXuZh,Tseng} proved that $\ell_1$-minimization gives a stable approximation of
the signal $\beta$ under some conditions that are stronger than (\ref{mip}).
The open problem of whether condition (\ref{mip}) , namely $\mu < \frac{1}{2k-1}$, is sufficient for
stable approximation of $\beta$ in the noisy case was settled by  Cai, Wang and Xu in \cite{CaWaXu}--actually the
authors of \cite{CaWaXu} even proved that this condition is sharp, for both noisy and noiseless cases. Another remark is that
if one considers only the noiseless case, the proof in \cite{CaWaXu} simplifies that of \cite{DonHuo,Fuchs,GrNi}

Although the sparse recovery condition (\ref{mip}) is rather strong, it is advantageous in checking whether a matrix meets the condition. Such a checking procedure
requires $O(N^2)$ steps which is computationally feasible. The MIC $\mu$ has been explicitly used in the design of compressed sensing matrices
by several works, e.g., \cite{CaJi,Dev,BoDiFoKo,LGGZ,xu}.

\subsection{RIP and Sparse Recovery}

We say that $\Phi$ satisfies the Restricted Isometry Property (RIP) of order $k$ and constant $\delta_k
\in [0,1)$ if
\begin{equation}\label{eq:con}
(1-\delta_k) \|x\|_2^2 \leq \|\Phi x\|_2^2 \leq (1+\delta_k) \|x\|_2^2
\end{equation}
holds for all $k$-sparse vector $x$ (see \cite{CanTao05}).
 In fact, (\ref{eq:con}) is equivalent to requiring that the
 Grammian matrices $\Phi_T^\top\Phi_T$ has all of its eigenvalues in $[1-\delta_k,
 1+\delta_k]$ for all $T$ with $|T|\leq k$, where $\Phi_T$ is the submatrix
of $\Phi$ whose columns are those with indexes in $T$. It has been shown that,  under various conditions on RIP constant $\delta_k$, such as $\delta_{k}< \frac{1}3$ (see
\cite{CZ}),
one can
 recover the $k$-sparse signal by solving $(P_1)$.

\subsection{Katz' Characters Sums Estimation and Its Improvement}

Let $\F_{q^n}$ be a finite field of order $q^n$ where $q$ is a prime power.
A multiplicative character $\chi$  is a homomorphism from the  multiplicative group $\left<\F_{q^n}^{\ast}, \cdot\right>$ to $\MS^1$ where $\MS^1=\{z\in \C: \|z\|=1\}$. By a trivial character we mean the function that sends every element of $\F_{q^n}^{\ast}$ to $1$.

Let $g$ be a primitive root of $\F_{q^n}$ (namely, $g$ generates the multiplicative group $\F_{q^n}^{\ast}$). For each element $u\in \F_{q^n}^{\ast}$, we define
the discrete logarithm of $u$ with respect to the base $g$ to be the non-negative integer $m$ ( $m < q^n-1$)  such that
$g^m = u$, and we write $\log_g u = m$.

Let $N= q^n - 1$. Then a nontrivial multiplicative character of $\F_{q^n}^{\ast}$ is of the form
\[
\chi_a (u)= e^{a\frac{2\pi  \ri \log_g u}N}
\]
where $1\le a \le N-1$.

For a nontrivial multiplicative
character, a celebrated theorem of Katz \cite{katz} concerns the magnitude of summation of the
character values over a special coset of $\F_{q}$ (as an additive subgroup of $\F_{q^n}$). More precisely,
Katz proved that for a nontrivial multiplicative
character $\chi_a$ of $\F_{q^n}^*$, and for an element $\alpha$ of $\F_{q^n}$ with $\F_{q^n}=\F_{q}(\alpha)$, there is an estimation
\begin{equation}\label{katz}
\left|\sum_{t\in \F_{q}} \chi_a(t-\alpha)\right| \le (n-1)\sqrt{q}.
\end{equation}

The method that Katz used in \cite{katz} to obtain the estimate (\ref{katz}) was a geometric one. In \cite{li}, Li presented an
arithmetic proof of (\ref{katz}) using the Riemann hypothesis for the projective line over finite fields. It is interesting to
note, in the case of quadratic extension of the fields, we can get a more precise estimation of Katz sum with
an elementary and transparent proof. This result will be a useful tool in our construction of special CS matrices.
Our improvement of Katz estimation  is stated as
\begin{thm}\label{2dim_sum}
Suppose that $\chi_a$ is a nontrivial multiplicative
character of $\F_{q^2}^*$. For an element $\alpha$ of $\F_{q^2}$ with $\F_{q^2}=\F_{q}(\alpha)$, we have
\begin{eqnarray}\label{katz2}
\left|\sum_{t\in \F_{q}} \chi_a(t-\alpha)\right| &=& \sqrt{q},\qquad \mbox{ if } (q-1)\nmid a, \label{eq:lemma1}\\
\sum_{t\in \F_{q}} \chi_a(t-\alpha) &=& -1,\qquad \mbox{ if } (q-1)\mid a.\label{eq:lemma2}
\end{eqnarray}
\end{thm}

\begin{proof}
To this end, we first prove that
\begin{eqnarray}
\sum_{t\in \F_q^*}\chi_a(t)=\begin{cases}
0, & q-1\nmid a\\
q-1, & q-1\mid a
\end{cases} .\label{eq:suma}
\end{eqnarray}
Let $g$ be a primitive root in $\FF_{q^2}$. For $0\leq m\leq q-2$, the fact that $(g^{m(q+1)})^{(q-1)}=g^{m(q^2-1)}=1$
implies
\[
\{ g^{m(q+1)} : 0\leq m\leq q-2\}=\FF_q^*,
\]
or, equivalently
\begin{equation}\label{eq:q1}
(q+1) \mid \log_g t \iff t\in \F_q^*.
\end{equation}
Therefore, (\ref{eq:suma}) follows.

Now assume that $(q-1)\mid a$. We have
\begin{eqnarray*}
0&=& \sum_{t\in \FF_{q^2}^*}\chi_a(t)= \sum_{t_1,t_2\in \FF_q, (t_1,t_2)\neq (0,0)}\chi_a(t_1+t_2\alpha)\\
&=&\sum_{t\in \FF_q^*}\chi_a(t)+\sum_{t_2\in \FF_q^*}\sum_{t_1\in \FF_q}\chi_a(t_1+t_2\alpha)
= \sum_{t\in \FF_q^*}\chi_a(t)+ \sum_{t_2\in \FF_q^*}\chi_a(t_2) \sum_{t_1\in \FF_q}\chi_a(t_2^{-1}t_1+\alpha)\\
&=&\sum_{t\in \FF_q^*}\chi_a(t)+ \sum_{t_2\in \FF_q^*}\chi_a(t_2) \sum_{t\in \FF_q}\chi_a(t+\alpha)\\
&=&\left(1+\sum_{t\in \FF_q}\chi_a(t+\alpha)\right)\sum_{t\in \FF_q^*}\chi_a(t)=
\left(1+\sum_{t\in \FF_q}\chi_a(t+\alpha)\right)(q-1),
\end{eqnarray*}
and hence $\sum_{t\in \FF_q}\chi_a(t+\alpha)=-1$.

Next we consider the case $(q-1)\nmid a$ and we shall prove (\ref{eq:lemma1}).  It is
not difficult to check that $\sds\frac{t_1-\alpha}{t_2-\alpha}\notin \F_{q}^*$ for $t_1,t_2\in \F_q$
and $t_1\neq t_2$.
In fact, if $\sds \frac{t_1-\alpha}{t_2-\alpha} = s\in \F_{q}$, then $s\neq 1$.
Therefore $\alpha = \sds \frac{s t_2-t_1}{s-1}\in \F_q$.
This is absurd as $\alpha\notin \F_q$. Therefore by (\ref{eq:q1}) we see that
$$
(q+1)\nmid \log_g\frac{t_1-\alpha}{t_2-\alpha},
$$
provided $t_1,t_2\in \FF_q$ and $t_1\neq t_2$. Now
\begin{eqnarray*}
\Bigabs{\sum_{t\in \FF_q}\chi_a(t-\alpha)}^2&=&\Bigabs{\sum_{t\in \FF_q}e^{a\frac{2\pi \ri \log_g(t-\alpha)}{q^2-1}}}^2=
\left(\sum_{t_1\in \FF_q}e^{a\frac{2\pi \ri \log_g(t_1-\alpha)}{q^2-1}}\right)\left(\sum_{t_2\in \FF_q}e^{a\frac{-2\pi \ri \log_g(t_2-\alpha)}{q^2-1}}\right)\\
&=&q+\sum_{t_1,t_2\in\FF_q, t_1\neq t_2}e^{a\frac{2\pi \ri \log_g\frac{t_1-\alpha}{t_2-\alpha}}{q^2-1}}=q+\sum_{k=1, q+1\nmid k}^{q^2-2}e^{a\frac{2\pi \ri k}{q^2-1}}\\
&=&q+\sum_{k=1}^{q^2-2}e^{a\frac{2\pi \ri k}{q^2-1}}-\sum_{j=1}^{q-2}e^{a\frac{2\pi \ri j(q+1)}{q^2-1}}=q+\sum_{k=0}^{q^2-2}e^{a\frac{2\pi \ri k}{q^2-1}}-\sum_{j=0}^{q-2}e^{a\frac{2\pi \ri j(q+1)}{q^2-1}}\\
&=& q,
\end{eqnarray*}
which yields
$$\bigabs{\sum_{t\in \FF_q}\chi_a(t-\alpha)}=\sqrt{q}.$$
\end{proof}

\section{The Main Results and Construction Procedure}
\label{sec:main_res}
In this section, we shall discuss construction of CS matrices by deterministically selecting set of
rows from the $N\times N$ Fourier matrix.

In the following discussion, we let $F_0, F_1,\ldots, F_{N-1}$ be
the rows of the $N\times N$ Fourier matrix ${\cal F}^{(N)}$, i.e.,
\[
{\cal F}^{(N)} = \begin{pmatrix}
F_{0}\\ F_{1} \\ \vdots\\F_{N-1}
\end{pmatrix}.
\]

Suppose that $M$ is a subset of $\{0,1,\ldots,N-1\}$, i.e., $M=\{m_0,m_1,\ldots,m_r\}\subset \{0,1,\ldots,N-1\}$.  Then
we can define the partial Fourier matrix associated with $M$ as
\[
{\cal F}^{(N)}_M := \begin{pmatrix}
F_{m_0}\\ F_{m_1} \\ \vdots\\F_{m_{r}}
\end{pmatrix}.
\]
Let $q$ be a prime power and let $\alpha\in \F_{q^n}$ be such that
\[
\F_{q^n} = \F_q(\alpha).
\]
Assume that $g$ is a generator of the cyclic group $\F_{q^n}^{\ast}$, we then have
\begin{thm} \label{thm:3.1}
Let $q$ be a prime power, and $n>1$ be a positive integer. Let $N=q^n-1$ and
 $$M=\{m = \log_g (t-\alpha): t\in \F_q\}.$$
Then the $q\times N$ matrix $\Phi=\frac{1}{\sqrt{q}}{\cal F}^{(N)}_M $
has MIC
\[
\mu \le \frac{n-1}{\sqrt{q}}.
\]
\end{thm}
\begin{proof}
For $0\le j, k \le N-1$ and $j\neq k$, the inner product of the $j$th and
$k$th columns of the matrix $\Phi$
is
\begin{eqnarray*}
\langle \Phi_j, \Phi_k\rangle &=& \frac{1}q \sum_{r=0}^{q-1}
e^{\frac{2\pi j m_r \ri}N}e^{-\frac{2\pi k m_r \ri}N}
=\frac{1}q \sum_{r=0}^{q-1}e^{(j-k)\frac{2\pi m_r \ri}N}\\
&=&\frac{1}q \sum_{t\in \F_q}e^{(j-k)\frac{2\pi \log_g (t-\alpha)\ri}N}
\end{eqnarray*}

Since $j\neq k$,  $\chi (u) = e^{(j-k)\frac{2\pi \log_g (u)\ri}N}$ defines
a nontrivial multiplicative character of $\F_{q^n}^{\ast}$. By the
Katz estimation (\ref{katz}) we have
\[
\mu = \max_{j\neq k}|\langle \Phi_j, \Phi_k \rangle|
=\frac{1}q \left|\sum_{t\in \F_q}\chi (t-x)\right|\le \frac{n-1}{\sqrt{q}}.
\]
\end{proof}

We would like to point out that this result was also mentioned in \cite{BoDiFoKo}.
However, Theorem \ref{thm:3.1} only produces a $q\times (q^n-1)$ matrix and hence the
 matrix size in Theorem \ref{thm:3.1} is quite restrictive.
Note that in DeVore's construction \cite{Dev}, the matrix size is
$p^2\times p^n$ where $p$ is a prime. But Theorem \ref{thm:3.1} cannot always be
used to produce matrices of size
$\Theta(p^2\times p^n)$, e.g., when $n$ is an odd number.

We can have
more flexibilities in choosing the size $N$ of the Fourier matrix ${\cal F}^{(N)}$. This
provides a larger set of parameters for constructing
partial Fourier matrices for
compressed sensing.
More specifically, we have the following:
\begin{thm} \label{thm:3.2}
Let $q=p^a$ where $p$ is a prime number. Suppose that $b$ is a positive integer with $b\mid a$.
Let $\sds N=\frac{q^n-1}{p^b-1}$
where $n>1$ is a  positive integer and let
$$
M \,=\,\,\{m_k=\log_g (t_k-\alpha) \pmod N\, :\, t_k\in \F_q, k=0,\ldots,q-1\}.
$$
Then the MIC for the $q\times N$ matrix $\Phi=\frac{1}{\sqrt{q}}{\cal F}^{(N)}_M$ satisfies
\[
\mu \le \frac{n-1}{\sqrt{q}}.
\]
\end{thm}
\begin{proof} To this end, we
first prove that $\#M=q$, i.e.,   $m_k\neq m_j$ if $k\neq j$.
 We assume that, for some $k\neq j$,
$m_k = m_j$ holds. Then
\[
 \log_g (t_k-\alpha) \equiv \log_g (t_j-\alpha) \pmod N.
\]
This implies that
\begin{equation}\label{thm3.2_eq:1}
\frac{t_k-\alpha}{t_j-\alpha} = g^{dN}
\end{equation}
for some integer $d$. Obviously $g^{dN}\neq 1$ as $t_k\neq t_j$.
Noting that  $(g^{dN})^{p^b-1}=\big(g^{q^n-1}\big)^d = 1$, we see that
\[
g^{dN}\in \F_{p^b}\subset \F_q.
\]
Solving (\ref{thm3.2_eq:1}) for $\alpha$, we get
\[
\alpha\, =\, \frac{g^{dN}t_j-t_k}{g^{dN}-1}\,\in\, \F_q,
\]
which implies that $\F_q(\alpha)=\F_q$. This is impossible as $n>1$.

Now we consider the  $q\times N$ matrix $\Phi$.
The inner product of the $j$th and
$k$th columns of the matrix $\Phi$ is
\begin{eqnarray*}
\langle \Phi_j, \Phi_k\rangle &=& \frac{1}q \sum_{m\in M}
e^{\frac{2\pi j m \ri}{N}}e^{-\frac{2\pi k m \ri}{N}}
=\frac{1}q \sum_{m\in M}e^{(j-k)(p^b-1)\frac{2\pi m \ri}{\tilde{N}}}\\
&=&\frac{1}q \sum_{t\in \F_q} e^{(j-k)(p^b-1)\frac{2\pi \log_g (t-\alpha)\ri}{\tilde{N}}},
\end{eqnarray*}
where $\tilde{N} = N(p^b-1) = q^n-1$.
Since for $0\le j, k < N$ and $j\neq k$,
\[
\chi (u) = e^{(j-k)(p^b-1)\frac{2\pi \log_g (u)\ri}{\tilde{ N}}}
\]
 defines a nontrivial multiplicative character of $\F_{q^n}^{\ast}$. Again, by the
Katz estimation (\ref{katz}), we arrive at
\[
\mu = \max_{j\neq k}|\langle \Phi_j, \Phi_k \rangle|
=\frac{1}q \left|\sum_{t\in \F_q}\chi (t-\alpha)\right|\le \frac{n-1}{\sqrt{q}}.
\]
\end{proof}

\begin{remark}
\begin{enumerate}
\item
According to our discussion in the previous section, a
$k$-sparse signal $\beta$ in model (\ref{P0}) can be reconstructed via
the $\ell_1$-minimization ($P_1$) as long as $\mu <\frac{1}{2k-1}$.
The matrix $\Phi$ in   Theorem \ref{thm:3.1} and Theorem \ref{thm:3.2},
can be a CS matrix for recovering $k$ sparse
signals if
\[
k < \frac{\sqrt{q}}{2(n-1)}+\frac{1}2.
\]
This bound is the same as the case of binary CS matrices obtained by DeVore \cite{Dev}.
Also, this bound improves the one for subsampling Fourier matrices in \cite{HaApNo} greatly.
\item In terms of matrix dimension, our theorems  provide
more flexibilities. In \cite{Dev}, the (binary) CS matrices can be of
size $p^2\times p^r$ for any $r\ge 2$. We can get similar dimensions
for partial Fourier CS matrices too. Taking $q=p^2$ in theorem \ref{thm:3.1},
we can get a $p^2\times \Theta(p^{2n})$ matrix; letting $a=2$ and $b=1$ in theorem \ref{thm:3.2},
we have
a $p^2\times \Theta(p^{2n-1})$ matrix.
\item
For the number ${N}=\frac{q^n-1}{p^b-1}=\frac{p^{an}-1}{p^b-1}$,
as pointed out in \cite{CaJi}, the coherence of a $q\times {N}$ random matrix with i.i.d.
Gaussian entries is about $2\sqrt{\frac{\log {N}}{q}}\thickapprox 2\sqrt{\frac{(an-b)\log p}{q}}$. Hence,
when $\log p>\frac{(n-1)^2}{4{(an-b)}}$, MIC of the deterministic  matrix $\Phi$ given in Theorem \ref{thm:3.2} is smaller
than that of random matrices.
\end{enumerate}
\end{remark}

Finally in this section, we restrict ourselves to the case of quadratic extension. We shall construct a partial Fourier matrix whose columns
form a union of orthonomal bases, by using our improvement of Katz' estimation. The construction of such matrix is also raised  in quantum information theory

Let $q$ be a prime power and let $\alpha$ be such that $\F_{q^2}=
\F_q(\alpha)$. As before, we let $N=q^2-1$ and denote
\begin{equation}\label{eq:M}
M=\{m = \log_g (t-\alpha):t\in \F_q\}.
\end{equation}
It is easy to see that $0\notin M$.  We are able to state
\begin{thm}\label{th:amub} Let
\begin{equation}\label{eq:phi}
\Phi = \frac{1}{\sqrt{q+1}} {\cal F}^{(N)}_{M\cup \{0\}}\qquad \in\qquad \C^{(q+1)\times (q^2-1)},
\end{equation}
where $M$ is defined in   (\ref{eq:M}).
For each $j=0,\ldots, q-2$, set
\[
T_j=\{j+k\cdot (q-1): 0 \leq k\leq q\}.
\]
Then we have
\begin{enumerate}
\item For any $0\leq j\leq q-2$, $\Phi_{T_j}$ is an orthogonal matrix.
\item For $k_1\in T_{j_1}$ and $k_2\in T_{j_2}$ with $j_1\neq j_2$, we have
\[
 \frac{\sqrt{q}-1}{q+1} \leq \bigabs{\left<\Phi_{k_1},\Phi_{k_2}\right>}\leq \frac{\sqrt{q}+1}{q+1}
\]
\end{enumerate}
\end{thm}
\begin{proof}
For any $k_1,k_2\in T_j$ with $k_1\neq k_2$. Since $(q-1)\mid k_1-k_2$, using
(\ref{eq:lemma2}) of theorem \ref{2dim_sum} we get
\begin{eqnarray*}
\langle \Phi_{k_1}, \Phi_{k_2}\rangle &=& \frac{1}{q+1}\left(1+ \sum_{r=0}^{q-1}
e^{\frac{2\pi k_1 m_r \ri}N}e^{-\frac{2\pi k_2 m_r \ri}N}\right)
=\frac{1}{q+1}\left(1+ \sum_{r=0}^{q-1}e^{(k_1-k_2)\frac{2\pi m_r \ri}N}\right)\\
&=&\frac{1}{q+1}(1+(-1))=0.
\end{eqnarray*}
This shows $\Phi_{T_j}$
is an orthogonomal matrix for any $0\le j \le q-2$.

We next consider the inner product of
$\Phi_{k_1}, \Phi_{k_2}$ with $k_1$ and $k_2$ belonging to different sets $T_j, j=0,\ldots,q-2$.
We have
\begin{eqnarray*}
\abs{\langle \Phi_{k_1}, \Phi_{k_2}\rangle} &= & \frac{1}{q+1}\Bigabs{1+ {\sum_{r=0}^{q-1}
e^{\frac{2\pi k_1 m_r \ri}N}e^{-\frac{2\pi k_2 m_r \ri}N}}}\\
&=&\frac{1}{q+1}\Bigabs{1+ {\sum_{r=0}^{q-1}
e^{\frac{2\pi (k_1-k_2) m_r \ri}N}}}
\end{eqnarray*}
Since $(q-1)\nmid k_1-k_2$, using
(\ref{eq:lemma1}) of theorem \ref{2dim_sum} we get
$$
\frac{\sqrt{q}-1}{q+1}\leq \abs{\langle \Phi_{k_1}, \Phi_{k_2}\rangle}\leq \frac{\sqrt{q}+1}{q+1}.
$$
\end{proof}
\begin{remark} The concept of
{\em approximately mutually unbiased bases } (AMUBs)
(see \cite{AMUB, AMUB1, AMUB2}) arises from quantum information theory. An
$m\times N$  matrix
$\Phi$ (with columns of unit length and with $m\mid N$) is a AMUBs if (i)
the set of columns of $\Phi$ can be partitioned into $\frac{N}m$ sets
of $m$ vectors, and each set is an orthonomal basis for $\C^m$; (ii)
$ \abs{\left<\Phi_{k_1},\Phi_{k_2}\right>}=\sqrt{\frac{1}{m}}+o(1)$ holds
if $\Phi_{k_1}$ and $\Phi_{k_2}$ are taken from different bases.
By Taylor expansion, we have
\begin{eqnarray*}
\frac{\sqrt{q}+1}{q+1}=\sqrt{\frac{1}{q+1}}+O\left(\frac{1}{q+1}\right),\qquad \frac{\sqrt{q}-1}{q+1}=\sqrt{\frac{1}{q+1}}+O\left(\frac{1}{q+1}\right).
\end{eqnarray*}
So Theorem  \ref{th:amub} says that we have actually constructed an AMUBs.
\end{remark}

\section{Computational Issues and Numerical Example}
In this section we will first discuss the issues of finding a primitive root $g$ in $\F_{q^n}$ and of computing discrete logarithm with respect to $g$.
We then explain some results from our numerical experiments.
\subsection{Computational Issues}
Finding a primitive root of a finite field is an interesting problem. In \cite{shpar},
Shparlinski gave a deterministic algorithm which returns a primitive root $g$ of $\F_{q^n}$
in time $O(q^{\frac{n}4+\varepsilon})$. It is noted that there is no polynomial time algorithm
available for this problem even under the General Riemann Hypothesis unless $q$ is prime and $n\le 2$ (see Shoup \cite{shoup}).
The idea of Shparlinski's algorithm is to construct a subset $M\subset \F_{q^n}$ such that
\begin{enumerate}
\item $M$ contains
a primitive root;
\item $|M| = O(q^{\frac{n}4})$.
\end{enumerate}
To identify a primitive root in $M$ one just uses
the following fact:  an element $g \in \F_{q^n}^{\ast}$ is a primitive root if and only if $g^{\frac{q^n-1}\ell}\neq \ell$
for any prime factor $\ell$ of $q^n-1$.

Computing a discrete logarithm over a finite field is also of great importance because the hardness of
this problem is a base of several well-known cryptosystems. The index calculus methods provide efficient ways
of solving discrete logarithm problem over  finite fields (see \cite{odly}). However, since in our setting, the size of
the finite field involved is far smaller than that in the cryptographical setting, we can use the following algorithm of Shanks \cite{shanks}
to compute the discrete logarithm in $\F_{q^n}^{\ast}$.  Again, denote $N=q^n-1$ and let $K=\lceil \sqrt{N}\rceil$.\\
\newpage
\begin{center}
\underline{\bf Computing $\log_g u$ for $u\in \F_{q^n}^{\ast}$}\\
\begin{tabular}{ll}
1& For all $0 \le j < K$, compute $g^j$ and store $(j, g^j)$ in list $L$.\\
2& Compute $g^{-K}$\\
3& For $ i = 0$ to $K-1$\\
 & $\quad\quad$ Compute $u (g^{-K})^i$ \\
 & $\quad\quad$ if $u (g^{-K})^i$ equals some $g^j$ in the list $L$, return $\log_g u =iK+j$.\\
\end{tabular}
\end{center}
The computational costs are $O(\sqrt{N})$ both in time and space.

\subsection{ Numerical Examples}
The purpose of the experiments is to provide  some performance comparisons
of the random sampling and the deterministic sampling of Fourier matrices.
We use the method introduced in this paper to produce a deterministic partial
Fourier matrix. To compare with the deterministic partial
Fourier matrix, we generate the random partial Fourier matrix by
choosing the same number of rows (as that in the deterministic case)
from the Fourier matrix in a uniformly
random  manner. We reconstruct the Fourier
coefficients by OMP (see \cite{rauhut1,xuomp}) and BP, respectively.
For BP, we use the optimization
tools of CVX \cite{cvx}.   We repeat the experiment 100 times for each
sparsity $k\in \{1,\ldots,20\}$ and calculate the success rate in Example 1
and Example 2. In Example 3, we compare the maximum and minimum eigenvalue
statistics of Gram matrices $\Phi_T^\top\Phi_T$ of varying sparsity
$k=|T|$ for a deterministic sampling matrix and a random sampling matrix.
It is interesting to note that all the numerical results show that
the deterministic partial Fourier matrix
has a better performance over the random partial Fourier matrix,
regardless
whether the sparsity if within the theoretical range (i.e. $k <
\frac{1}2\left(\frac{1}{\mu}+1\right)$) or not.

\begin{example}
Take $q=29$ and $N=q^2-1=840$. Consider the field $\F_{29^2}=\F_{29}[x]/(x^2+2)$.
Denote $\bar{x} = x + (x^2+2)$ and choose $\alpha=\bar{x}, g=\bar{x}+2$. We have
\begin{eqnarray*}
M_1 &=& \{\log_g(t-\alpha):t\in \F_{29}\}\\
&=&\{465, 1, 494, 649, 47, 507, 758, 610, 835, 244, 67, 204, 588, 519, \\
& & \: \: 332, 808,
351,  672, 456, 683, 776, 275, 470, 562, 3, 103, 761, 466, 449\},
\end{eqnarray*}
By Theorem \ref{thm:3.1}, we can build a $29\times 840$ partial
Fourier matrix $\Phi=\frac{1}{\sqrt{q}}{\cal F}^{(N)}_{M_1}$.
To compare, we generate the random partial Fourier matrix
selecting $29$ rows from the $840\times 840$ Fourier matrix
uniformly randomly.

 Figure 1  depicts the recovery results with using the recovery  algorithm OMP and BP, respectively.
\end{example}

  \begin{figure}[!ht]
\begin{center}
\epsfxsize=5cm\epsfbox{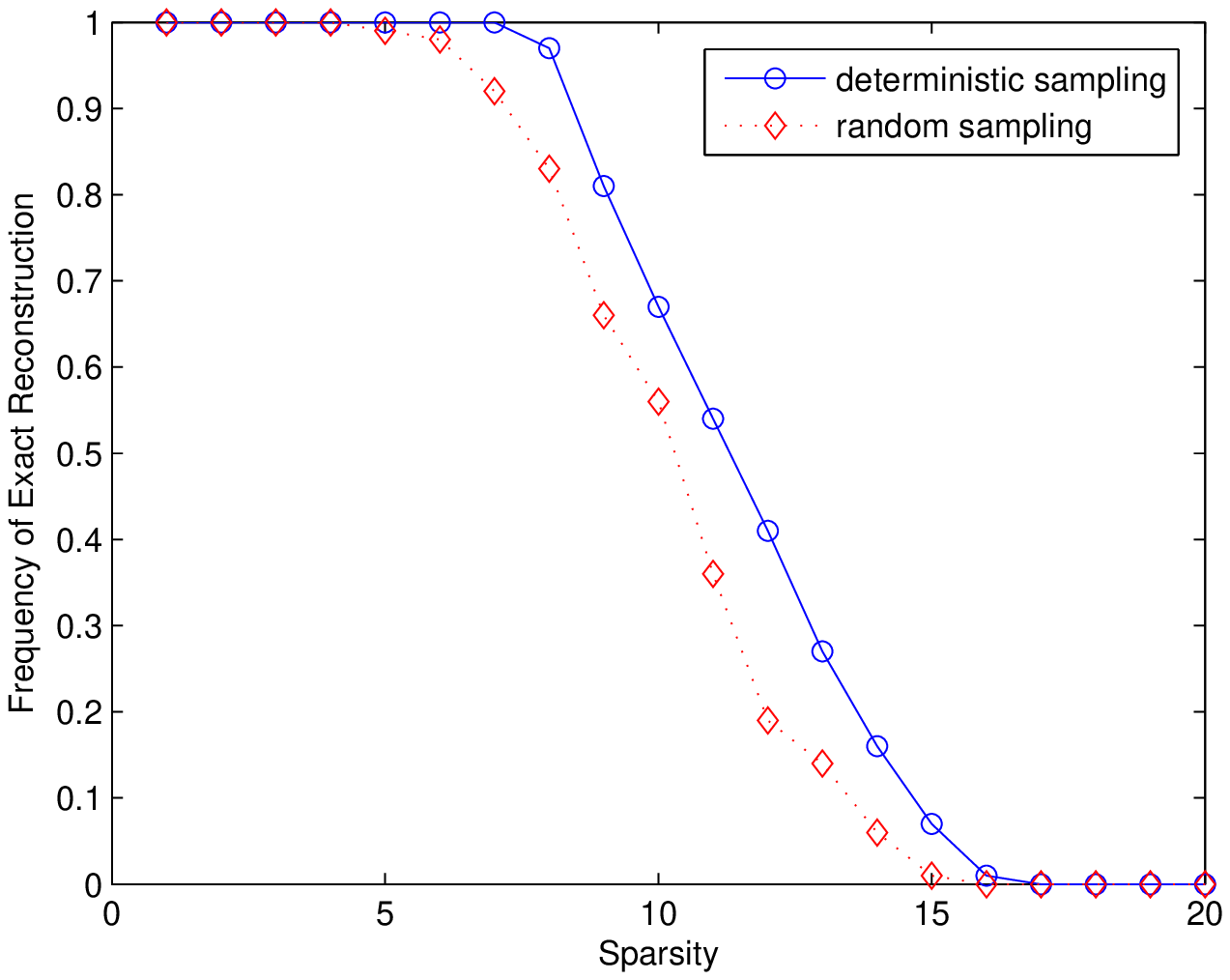}\,\,\,\,\,\,\,\,\, \epsfxsize=5cm\epsfbox{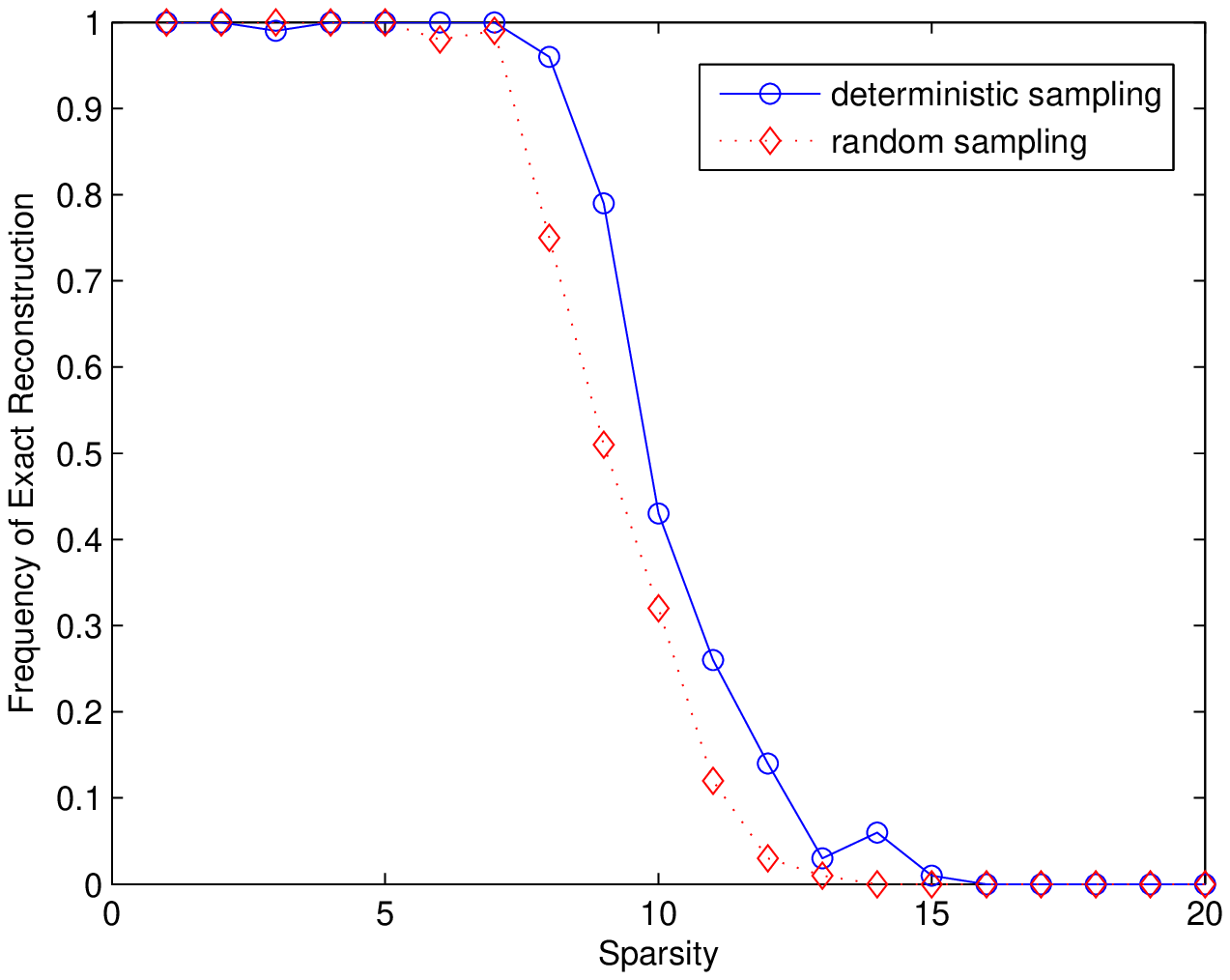}
\end{center}
\caption{
Numerical experiments for the  comparison of the random  and the deterministic  $29\times 840$ partial Fourier matrix.  The left graph corresponds to the success rates recovering by OMP, whereas the right one depicts the success rate of BP.}

\end{figure}

\begin{example}
We take $q=p=19$ and ${N}=\frac{p^3-1}{p-1}=p^2+p+1=381$. Then we  have
\begin{eqnarray*}
M_2 &=& \{\log_g(t-\alpha)\: {\rm mod }\: {N} :t\in \F_{19}\}\\
&=&\{ 192,   208,   162,   165,   160,    39,   154,   141,   245,   356,   304,   311,   223,   40,   321,    68,   118,   174,   249\},
\end{eqnarray*}
where $\F_{19^3}=\F_{19}[x]/(x^3+x+1), \alpha=\bar{x}, g=\bar{x}^2+2\bar{x}$.
By Theorem \ref{thm:3.2}, we can build a $19\times 381$ sub-Fourier matrix $\Phi=\frac{1}{\sqrt{q}}{\cal F}^{({N})}_{M_2}$. Similar to Example 1, we also compare the deterministic sampling and the random sampling and show the result in Figure 2.
  \begin{figure}[!ht]
\begin{center}
\epsfxsize=5cm\epsfbox{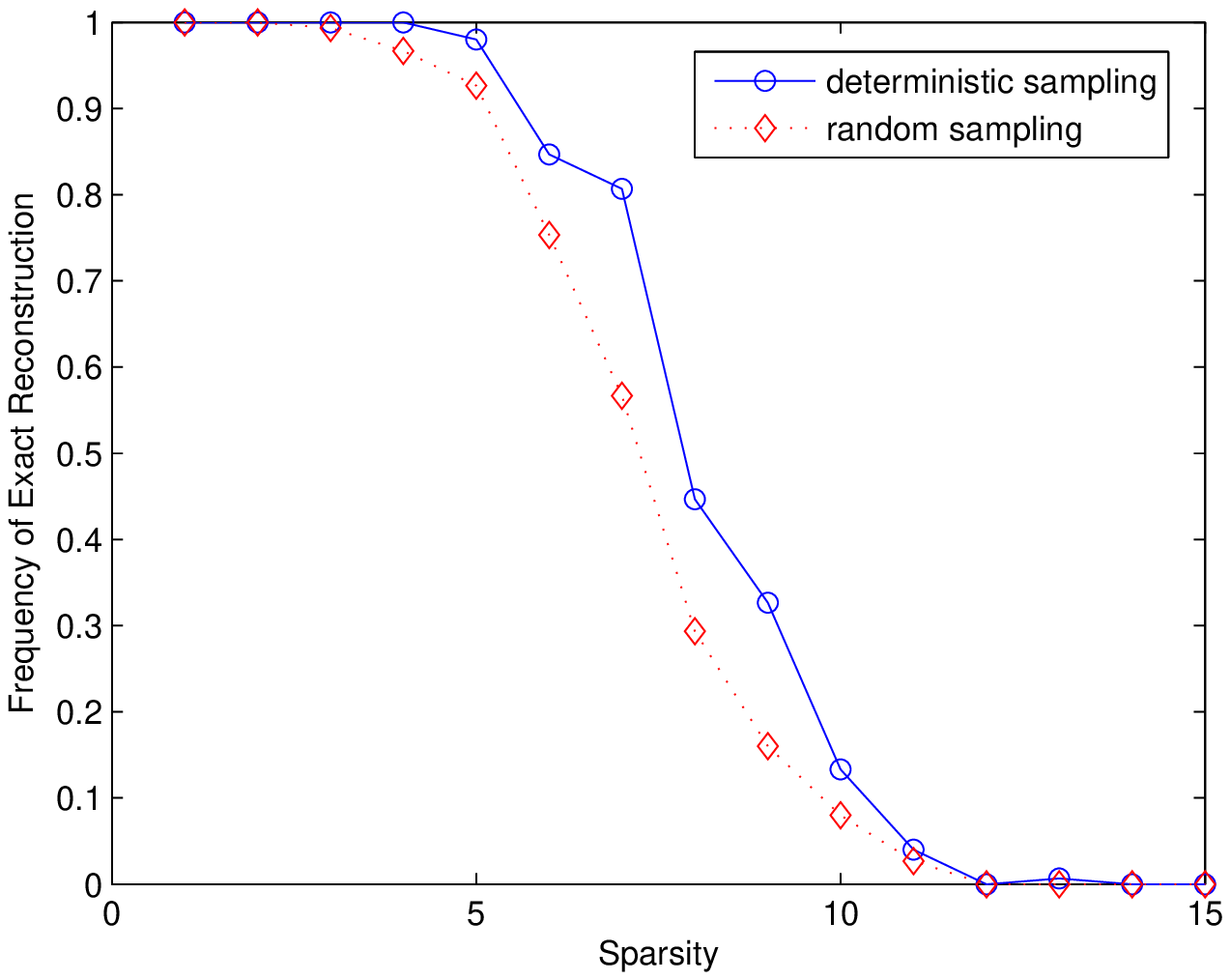}\,\,\,\,\,\,\,\,\, \epsfxsize=5cm\epsfbox{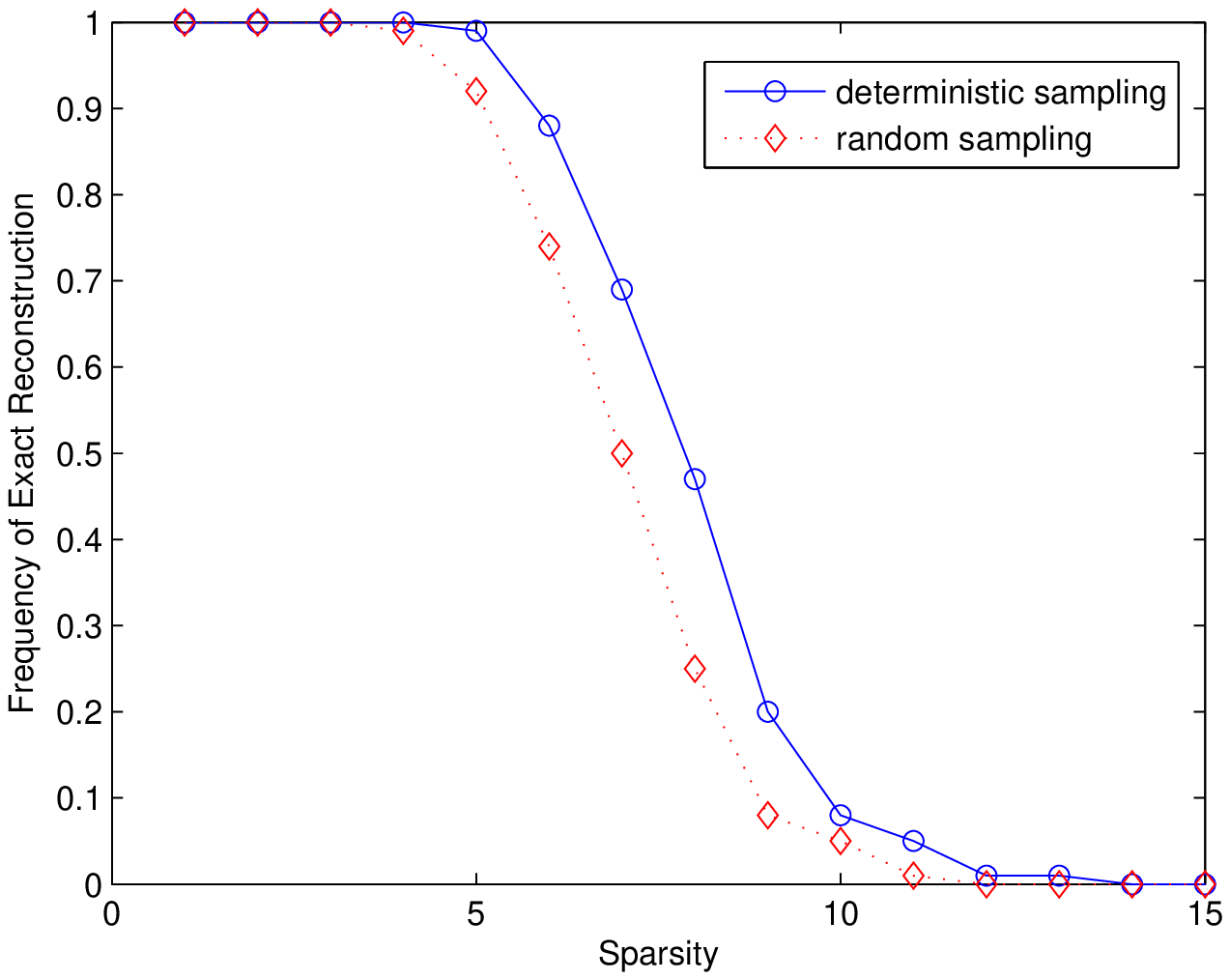}
\end{center}
\caption{Numerical experiments for the  comparison of the random   and the deterministic  $19\times 381$ partial Fourier  matrix.  The left graph corresponds to the success rates recovered by OMP, whereas the right one depicts the success rate of BP.}

\end{figure}
\end{example}

\begin{example}
The aim of the numerical experiment is to compare the maximum and minimum eigenvalue statistics of Gram matrices $\Phi_T^\top\Phi_T$ of varying sparsity $M=|T|$ for deterministic sampling matrix and random sampling matrix.
In fact,  RIP of order $k$ with constant $\delta_k$ is equivalent to requiring that the
 Grammian matrices $\Phi_T^\top\Phi_T$ has all of its eigenvalues in $[1-\delta_k,  1+\delta_k]$ for all $T$ with $|T|\leq k $.
For every value $M$, sets $T$ are drawn uniformly random over all sets and the statistics are accumulated
  from 50,000 samples. Figure 3 shows  the maximum and minimum eigenvalues of $\Phi_T^\top\Phi_T$ for  $\#T\in \{1,\ldots,20\}$.
   \begin{figure}
\begin{center}
\epsfxsize=5cm\epsfbox{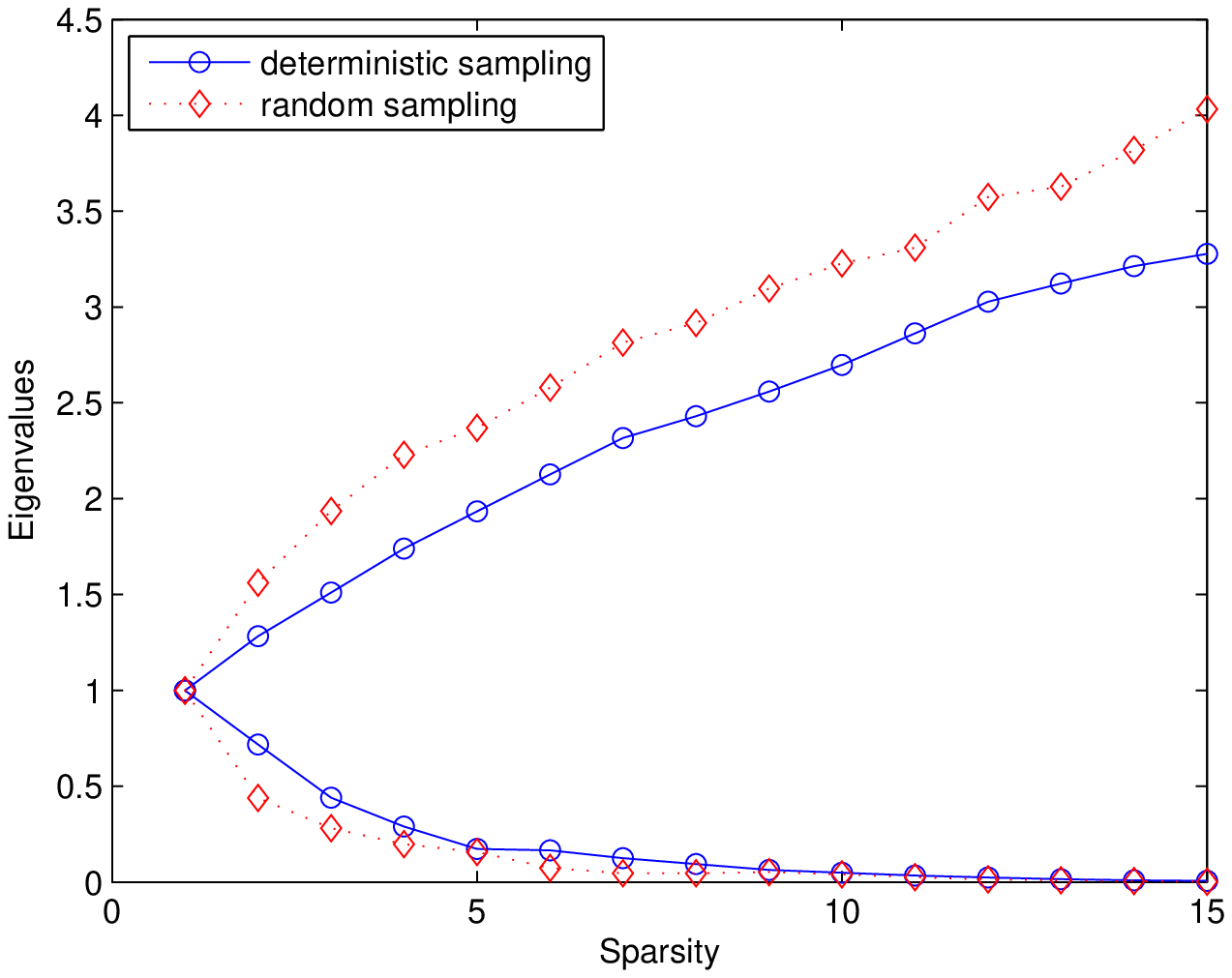} \,\,\,\,\,\,\,\,\, \epsfxsize=5cm\epsfbox{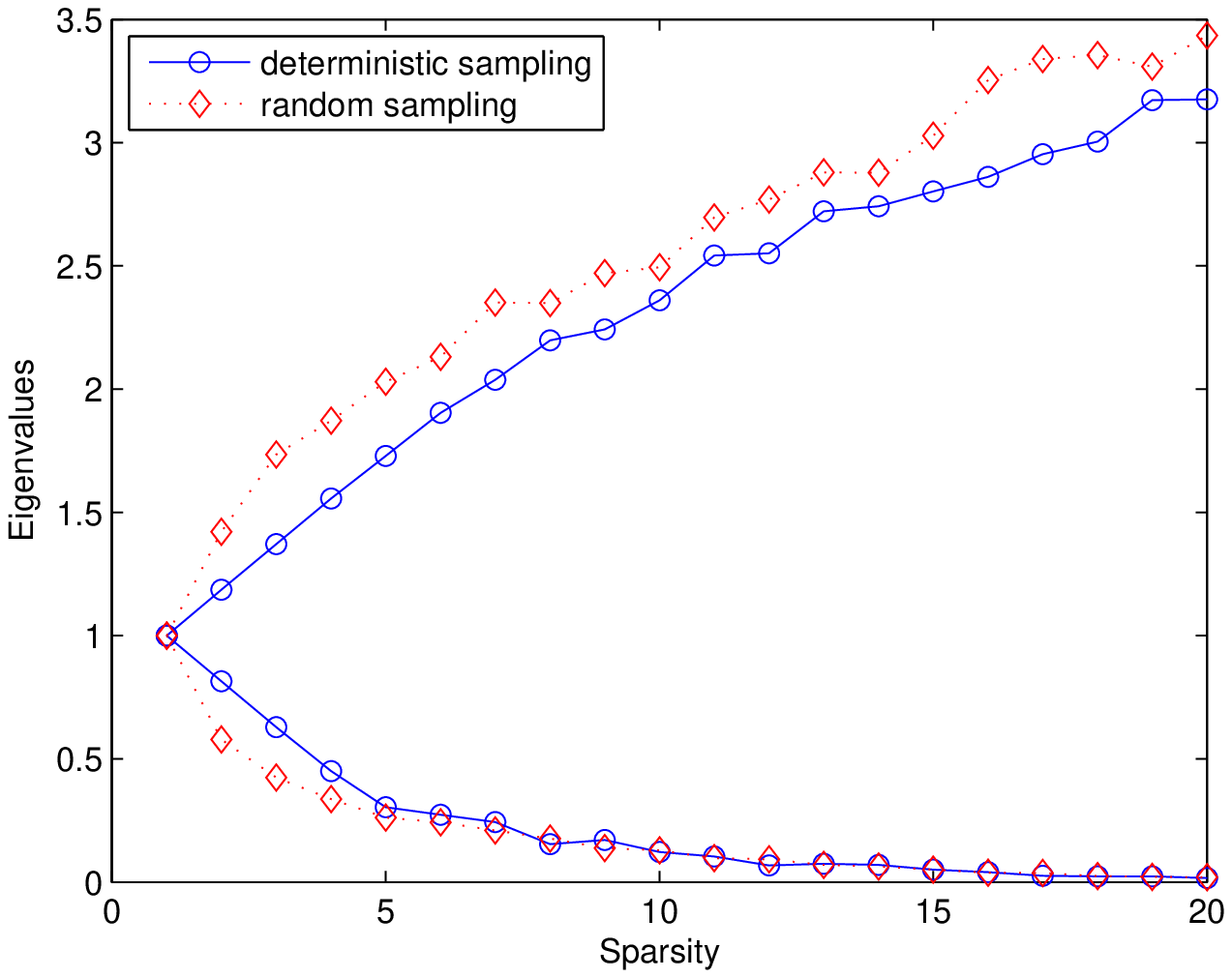}\caption{Eigenvalue statistics of
Gram matrices $\Phi_T^\top\Phi_T$ for the deterministic sampling matrixes
and random sampling matrixes. The left graph shows the result of
$19\times 381$ partial Fourier matrix in Example 2, whereas the right one
depicts the result of $29\times 840$ partial Fourier matrix in Example 1.}
\end{center}
\end{figure}
\end{example}

{\bf Acknowledgement}:
Part of this work was done when the first author visited the
Inst. Comp. Math., Academy of Mathematics and Systems Science,
Chinese Academy of Sciences. He is grateful for the warm hospitality.

\end{document}